\documentclass[11pt]{article}
\usepackage{verbatim,url,enumerate,color,paralist}
\usepackage{amsmath,amsfonts}
\usepackage{epsfig,amssymb,amstext,xspace,theorem}
\usepackage{algorithm}
\usepackage{algorithmicx,algpseudocode}
\usepackage{fullpage}
\usepackage[margin=1in]{geometry}

\newcommand{\qed}{\hspace*{\fill}$\Box$}
\newtheorem{theorem}{Theorem}[section]
\newtheorem{lemma}[theorem]{Lemma}
\newtheorem{corollary}[theorem]{Corollary}
\newtheorem{definition}[theorem]{Definition}

\newtheorem{claim}[theorem]{Claim}

\newenvironment{proof}[1][Proof. ]{\noindent {\bf #1 }}{\qed}

\newcommand{\Xomit}[1]{}

\begin{document}

\title{Multiple Traveling Salesmen in Asymmetric Metrics}

\author{Zachary Friggstad\footnote{Research supported by an iCORE ICT/AITF award while studying at the University of Alberta.} \\
Department of Combinatorics and Optimization\\ University of Waterloo\\ \texttt{zfriggstad@math.uwaterloo.ca}}

\begin{titlepage}
\maketitle
\thispagestyle{empty}

\begin{abstract}
We consider some generalizations of the Asymmetric Traveling Salesman Path problem. Suppose we have an asymmetric metric $G = (V,A)$ with two distinguished nodes $s,t \in V$. We are also given a positive integer $k$. The goal is to find $k$ paths of minimum total cost from $s$ to $t$ whose union spans all nodes. We call this the $k$-Person Asymmetric Traveling Salesmen Path problem ($k$-ATSPP). Our main result for $k$-ATSPP is a bicriteria approximation that, for some parameter $b \geq 1$ we may choose, finds between $k$ and $k + \frac{k}{b}$ paths of total length $O(b \log |V|)$ times the optimum value of an LP relaxation based on the Held-Karp relaxation for the Traveling Salesman problem. On one extreme this is an $O(\log |V|)$-approximation that uses up to $2k$ paths and on the other it is an $O(k \log |V|)$-approximation that uses exactly $k$ paths.

Next, we consider the case where we have $k$ pairs of nodes $\{(s_i, t_i)\}_{i=1}^k$. The goal is to find an $s_i-t_i$ path for every pair such that each node of $G$ lies on at least one of these paths. Simple approximation algorithms are presented for the special cases where the metric is symmetric or where $s_i = t_i$ for each $i$. We also show that the problem can be approximated within a factor $O(\log n)$ when $k=2$. On the other hand, we demonstrate that the general problem cannot be approximated within any bounded ratio unless P = NP.
\end{abstract}

\end{titlepage}

\section{Introduction}

We consider generalizations of the metric Traveling Salesman problem. In most of our settings we are given a complete directed graph $G = (V,A)$ with $n$ nodes. There are nonnegative arc costs $d_{uv}, uv \in A$ satisfying the directed triangle inequality $d_{uv} \leq d_{uw} + d_{wv}$ for any distinct $u,v,w \in V$. In general, $d_{uv} \neq d_{vu}$ for $u,v \in V$ so we call such graphs {\em asymmetric metrics}. Some well-studied Traveling Salesman variants in asymmetric metrics are to find the minimum cost Hamiltonian cycle or minimum cost Hamiltonian path where some of the endpoints of the path may be specified in advance. All of these problems are NP-hard to approximate within small constant factors \cite{papadimitriou:vempala}.

We will mainly study problems concerning finding multiple paths in asymmetric metrics whose union covers all nodes while minimizing the total cost of these paths. This generalizes the Asymmetric Traveling Salesman Path problem (ATSPP). Formally, define $k$-ATSPP to be the following problem. Given an asymmetric metric $G = (V,A)$ with arc costs $d_{uv}$ and two distinct nodes $s,t \in V$, we want to find $k$ paths from $s$ to $t$ in $G$ such that every node $v \in V$ lies on at at least one of these paths. Finally, we define General $k$-ATSPP to be the following generalization of $k$-ATSPP. We are given $k$ pairs of nodes $(s_i,t_i), \ldots, (s_k,t_k)$ in an asymmetric metric $G$ and we are to find an $s_i-t_i$ path for each $i$ so each $v \in V$ lies on at least one of these paths. Again, the goal is find such paths of minimum total cost.

One thing that makes $k$-ATSPP an attractive variant of ATSPP is that the gap between optimum solutions for different values of $k$ in an asymmetric metric may be arbitrarily large. For example, the instance in Figure \ref{fig:katspp_diff} has a solution of cost 0 using 2 paths but any single path has cost at least 1. One way to think about this is that it might be efficient to hire additional salesmen to cover the locations in an asymmetric metric. On the other hand, we do not have this large gap in symmetric metrics ({\em i.e.} when $d_{uv} = d_{vu}$ for all $u,v \in V$) because a single salesman can cover all paths by traveling back and forth between $s$ and $t$ and cover all locations with no greater cost (if $k$ is even then one final step from $s$ to $t$ makes it an $s-t$ path while adding an extra $OPT/k$ to the cost).

\begin{figure}
\centering
\includegraphics[scale=0.6]{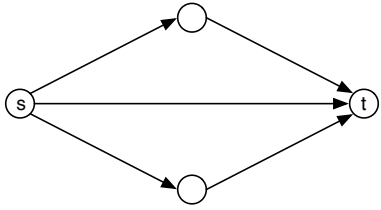}
\caption{
The pictured arcs all have cost 0 and the missing arcs have cost 1. The gap between optimum solutions for $k=1$ and $k=2$ is unbounded.
}
\label{fig:katspp_diff}
\end{figure}

\subsection{Related Work}
In the well-studied Traveling Salesman problem (TSP), the goal is to find a Hamiltonian cycle of minimum total edge cost in a symmetric metric. A classic result of Christofides \cite{christofides} is a polynomial-time algorithm for TSP that finds a Hamiltonian cycle with cost at most $\frac{3}{2}$ times the cost of the optimum solution. Hoogeveen \cite{hoogeveen} adapted this algorithm to the problem of finding minimum cost Hamiltonian paths. He obtains a $\frac{3}{2}$-approximation if at most one endpoint is fixed in advance and a $\frac{5}{3}$-approximation if both endpoints are fixed in advance. Recently, An, Kleinberg, and Shmoys have improved the approximability of the case when both endpoints are fixed to $\frac{1+\sqrt{5}}{2} < 1.6181$ \cite{an:kleinberg:shmoys}.

In asymmetric metrics, Frieze, Galbiati, and Maffioli \cite{frieze:galbiati:maffioli} gave the first approximation algorithm for ATSP with an approximation ratio of $\log_2 n$ where $n = |V|$. A series of papers improved on this ratio by constant factors \cite{blaser, kaplan, feige:singh} with the last being $\frac{2}{3} \log_2 n$. Finally, Asadpour {\em et al.} \cite{asadpour:atsp} produced an asymptotically better approximation algorithm for ATSP with ratio $O(\log n / \log\log n)$.

The variant of finding Hamiltonian paths in asymmetric metrics, namely ATSPP, has only recently been studied from the perspective of approximation algorithms. The first approximation algorithm was an $O(\sqrt n)$-approximation by Lam and Newman \cite{lam:newman}. Following this, Chekuri and Pal \cite{chekuri:pal:atspp} brought the ratio down to $O(\log n)$. Finally, Feige and Singh \cite{feige:singh} proved that an $\alpha$-approximation for ATSP implies a $(2 + \epsilon)\alpha$-approximation for ATSPP for any constant $\epsilon > 0$. Combining their result with the recent ATSP algorithm in \cite{asadpour:atsp} yields an $O(\log n / \log\log n)$-approximation for ATSPP.

There is a linear programming (LP) relaxation for each of these problems based on the Held-Karp relaxation for TSP \cite{held:karp}. For TSP, this relaxation is:
\[
\begin{array}{rrcll}
{\rm minimize}: & \sum_{uv \in E} d_{uv} x_{uv} & & & \\
{\rm such~that}: & x(\delta(v)) &=& 2& \forall~ v \in V \\
 & x(\delta(S)) &\geq & 2 & \forall~ \emptyset \subsetneq S \subsetneq V \\
 & x_{uv} &\in & [0,1] & \forall~ uv \in E
\end{array}
\]
Many of the approximation algorithms mentioned above also bound the integrality gap of the respective Held-Karp LP relaxation. For TSP, Wolsey \cite{wolsey:heuristic} proved the solutions found by Christofides' algorithm \cite{christofides} are within 3/2 of the optimal solution to the above LP relaxation. For ATSP, Williamson \cite{williamson:msthesis} proved that the algorithm of Frieze {\em et al.} \cite{frieze:galbiati:maffioli} bounds the integrality gap of its respective LP by $\log_2 n$. The improved $O(\log n / \log\log n)$-approximation for ATSP in \cite{asadpour:atsp} improved the bound on gap to the same ratio. For TSP paths, An, Kleinberg and Shmoys \cite{an:kleinberg:shmoys} first showed that Hoogeveen's algorithm bounds the integrality gap of a Held-Karp type relaxation for TSP paths by $\frac{5}{3}$ in cases where both endpoints are fixed. In the same paper they argue that their $\frac{1+\sqrt 5}{2}$-approximation for this case also bounds the integrality gap by the same factor.

Nagarajan and Ravi \cite{nagarajan:ravi:latency} first showed that the integrality gap of an LP relaxation for ATSPP, which is the same as LP (\ref{lp:katspp}) in this paper when $k=1$, was $O(\sqrt n)$. Later Friggstad, Salavatipour, and Svitkina \cite{friggstad:salavatipour:svitkina} showed a bound of $O(\log n)$ in the integrality gap of this LP relaxation which is currently the best bound. We note that the result of Feige and Singh in \cite{feige:singh} that relates the approximability of ATSP and ATSPP does not extend to their integrality gaps in any obvious way.

In the full version of \cite{friggstad:salavatipour:svitkina}, the authors studied extensions of their $O(\log n)$-approximation for ATSPP to $k$-ATSPP. They demonstrated that $k$-ATSPP can be approximated within $O(k^2 \log n)$ and that this bounds the integrality gap of LP (\ref{lp:katspp}) by the same factor. Though not stated explicitly, their techniques can also be used to devise a bicriteria approximation for $k$-ATSPP that uses $O(k \log n)$ paths of total cost at most $O(k \log n)$ times the value of LP (\ref{lp:katspp}) in a manner similar to the algorithm in the proof of Theorem 1.3 in \cite{friggstad:salavatipour:svitkina}. As far as we know, no results are known for General $k$-ATSPP even for the case $k=2$.

One other problem related to one we consider is the following. We are given $2k$ distinct nodes  $S = \{s_1, \ldots, s_k\}$ and $T = \{t_1, \ldots, t_k\}$ in a symmetric metric. We want to find $k$ paths whose union spans all nodes. This should be such that each node in $S$ is the start node of exacly one path and each node in $T$ is the end node of exactly one path. Matroid intersection techniques used by Rathinam and Sengupta \cite{rathinam:sengupta} can be easily adapted to approximate this problem within a factor 2.

\subsection{Our Results}

By the directed triangle inequality, it is easy to see that there is an optimum solution for an instance of $k$-ATSPP where each node in $V-\{s,t\}$ lies on precisely one of the $k$ paths. Such an optimum solution corresponds to an integer point in LP (\ref{lp:katspp}) of the same cost. So the optimum value of LP (\ref{lp:katspp}), say $OPT_{LP}$, is a lower bound for the minimum cost $k$-ATSPP solution. Our main result for $k$-ATSPP is the following.
\begin{theorem} \label{thm:katspp}
For any integer $b \geq 1$, there is an efficient algorithm for $k$-ATSPP that finds between $k$ and $k + \frac{k}{b}$ paths of total cost at most $O(b \log n) \cdot OPT_{LP}$.
\end{theorem}
This is an $O(k \log n)$-approximation using precisely $k$ paths when $b=k+1$ and an $O(\log n)$-approximation using at most $2k$ paths when $b=1$.

The algorithm is also easy to implement with the most complicated subroutine being that of finding a minimum weight perfect matching in a bipartite graph. Its running time can easily be seen to be $O(M(n+k-2) b \log n)$ where $M(N)$ is the time it takes to compute such a matching in a complete bipartite graph with $N$ nodes on each side.

We then proceed to study variants of $k$-ATSPP that vary how the start and/or end locations are specified. Examples are when the start locations are not fixed or when we have a set of $k$ start nodes $S$ and a set of $k$ end nodes $T$ and the start and end locations of the paths should establish a bijection between $S$ and $T$. We extend our approximation algorithm for $k$-ATSPP to these variants.

Finally, we study General $k$-ATSPP.  Our first result is an $O(\log n)$ for General 2-ATSPP. We also have a 3-approximation for General $k$-ATSPP in symmetric metrics and an $O(\log n)$-approximation for General $k$-ATSPP when $s_i = t_i$ for all $i$. However, we have the following hardness result for General $k$-ATSPP with no further restrictions.
\begin{theorem} \label{thm:general_hard}
It is NP-hard to distinguish between instances of General $k$-ATSPP whose optimum solution has cost 0 and instances whose optimum solution has cost at least 1.
\end{theorem}
This implies the problem cannot be efficiently approximated within any bounded ratio unless P = NP.  While the reduction uses $k = n/4$, modifications can be made to prove hardness results (under stronger assumptions) for values of $k$ being as small as polylogarithmic in $n$.

To summarize, Section \ref{sec:katspp} presents the algorithm for $k$-ATSPP, proves Theorem \ref{thm:katspp}, and discusses some variations of $k$-ATSPP on how the start and/or end locations are specified. In Section \ref{sec:general_katspp} we demonstrate an $O(\log n)$-approximation for General 2-ATSPP, discuss approximation algorithms for other restrictions of General $k$-ATSPP, and prove Theorem \ref{thm:general_hard}. Section \ref{sec:conc} then concludes this paper by identifying some directions for future work and mentioning some basic results for the alternative goal of minimizing the cost of the most expensive path in either $k$-ATSPP or General $k$-ATSPP.


\section{A Bicriteria Approximation for $k$-ATSPP} \label{sec:katspp}

In this section, we will develop a bicriteria approximation algorithm that finds approximately $k$ paths from $s$ to $t$ in an asymmetric metric $G = (V,A)$ whose total cost is within some bounded ratio of the optimum value of LP relaxation (\ref{lp:katspp}). The algorithm is parameterized by a positive integer $b$; different bicriteria approximation guarantees result from different choices of $b$.

\subsection{Preliminaries}

If $X$ is a flow between two nodes or a circulation then we let $X_{uv}$ denote the value that $X$ assignes to arc $uv \in A$. For $S \subseteq V$ we let $X(\delta^+(S)) = \sum_{u \in S, v \in V-S} X_{uv}$ and $X(\delta^-(S)) = \sum_{v \in V-S, u \in S} X_{vu}$. For brevity, we let $X(\delta^+(v)) := X(\delta^+(\{v\}))$ and $X(\delta^-(v)) := X(\delta^-(\{v\}))$ for $v \in V$. We say $X$ is integral if $X_{uv}$ is an integer for each arc $uv \in A$. The cost of $X$ is $\sum_{uv \in A} d_{uv} \cdot X_{uv}$. All flows and circulations $X$ in this paper will have $X_{uv} \geq 0$ for each arc $uv \in A$.

Our starting point will be to use structures similar to cycle covers from \cite{frieze:galbiati:maffioli} and path/cycle covers from \cite{lam:newman} and \cite{friggstad:salavatipour:svitkina}. 

\begin{definition}
A $k$-path/cycle cover from $s$ to $t$ is an integral flow $F$ such that $F(\delta^+(v)) = F(\delta^-(v)) = 1$
for each $v \in V-\{s,t\}$, $F(\delta^+(s)) = F(\delta^-(t)) = k$, and $F(\delta^-(s)) = F(\delta^+(t)) = 0$.
\end{definition}
Note that in a $k$-path/cycle cover $F$ the flow $F_{uv}$ across any arc $uv \in A-\{st\}$ is either 0 or 1 and $F_{st} \leq k$. If we regard $F$ as a multiset of arcs, then $F$ may be decomposed into $k$ paths from $s$ to $t$ and a collection of cycles where every $v \in V-\{s,t\}$ lies on exacly one path or exactly one cycle. We can efficiently find a minimum-cost $k$-path/cycle cover using a standard reduction to minimum weight perfect matching in a bipartite graph with $n+k-2$ nodes on each side.

LP (\ref{lp:katspp}) is the LP relaxation for $k$-ATSPP we consider. It is similar to the LP relaxation for ATSPP considered in \cite{nagarajan:ravi:latency} and \cite{friggstad:salavatipour:svitkina}.
\begin{align}
{\rm minimize:~} & \hspace{20mm} \sum_{e \in A} d_{uv} x_{uv} & & & \label{lp:katspp} \\
{\rm subject~to:~} & \hspace{10.5mm} x(\delta^+(v)) = x(\delta^-(v) = 1 & \forall v \in V-\{s,t\} \notag \\
 & \hspace{10mm} x(\delta^+(s)) = x(\delta^-(t)) = k & \notag \\
 & \hspace{10mm} x(\delta^-(s)) = x(\delta^+(t)) = 0 & \notag \\
 & \hspace{28.25mm} x(\delta^+(S)) \geq 1 & \forall \{s\} \subseteq S \subsetneq V \label{lp:katspp-set} \\
 & \hspace{37.75mm} x_{uv} \geq 0 & \forall~ uv \in A \notag
\end{align}

We will break the presentation of the algorithm into two parts. The first is a loop that runs for $\Theta(b \log n)$ iterations. Each iteration will find the cheapest $k$-path/cycle cover on the remaining nodes and discards some nodes from the cycles or, more generally, circulations formed by taking the union of the current $k$-path/cycle cover and paths from previous iterations. These discarded nodes can be added to the final solution by using the circulations to ``graft'' them into the paths. It is similar to the main loop in the algorithm for ATSPP in \cite{friggstad:salavatipour:svitkina}.

After this first phase we will have $\Theta(b k \log n)$ paths of cost from $s$ to $t$ whose union is acyclic and covers all remaining nodes. However, our goal is to use at most $k + \frac{k}{b}$ paths. The second part of the algorithm will assemble only $k + \frac{k}{b}$ paths that cover all remaining nodes using arcs from the $\Theta(b k \log n)$ paths. This will be possible because we will carefully select which nodes to discard in each iteration so that each remaining node lies in almost a $\frac{1}{k}$-fraction of the $\Theta(b k \log n)$ paths after the main loop. So, if we regard our $\Theta(bk\log n)$ paths as a flow then each node supports almost a $\frac{1}{k}$-fraction of this flow. If we scale the flow by $\Theta(b \log n)$ then we have a flow sending $\Theta(k)$ units from $s$ to $t$ where every other remaining node supports $\Omega(1)$ units of this flow.

We will show how to round this flow to obtain an acyclic integral flow of roughly the same cost that sends between $k$ and $k + \frac{k}{b}$ units of flow from $s$ to $t$ so that each remaining node supports exactly 1 unit of this flow. Then a path decomposition of this acyclic integral flow yields the required paths. Finally, we graft the nodes that were discarded in the first phase to these paths using the circulations that were removed in the first phase.

\subsection{The Algorithm}
Let $b \geq 1$ be an integer, this is the $b$ in the statement of Theorem \ref{thm:katspp}. For notational convenience, we will let $L$ be $(b+1)\lfloor \log_2 n \rfloor$ for the remainder of this section. Algorithm \ref{alg:katspp} is the $k$-ATSPP bicriteria approximation.

\begin{algorithm}[ht]
  \caption{~~The Bicriteria Approximation for $k$-ATSPP.} \label{alg:katspp} 
\begin{algorithmic}[1] 
\State Let $W \leftarrow V$ \Comment $W$ is the set of nodes that have not been discarded yet
\State Let $l_v = 0, \forall v \in V$ \Comment Used to decide which nodes to discard from a circulation
\State Let $F_{uv} \leftarrow 0, \forall u,v \in V$ \Comment $F$ will be an acyclic $s-t$ flow using only nodes in $W$
\State Let $C_{uv} \leftarrow 0, \forall uv \in V$ \Comment $C$ will be a circulation on the nodes in $V-W$
\State $L \leftarrow (b+1) \lfloor \log_2 n \rfloor$
\For   {$L$ iterations}
\State Find a minimum-cost $k$-path/cycle cover $F'$ on $W$ \label{step:kpathcycle}
\State $F \leftarrow F + F'$ \label{step:adding}
\State Subtract a circulation $H$ from $F$ so $F$ is acyclic again \label{step:circulation}
\For   {each connected component $A$ in the support of $H$} \label{step:innerloop}
\State Let $v_A \leftarrow \arg\min_{u \in A} l_u + H(\delta^+_A(u))$ \Comment breaking ties arbitrarily \label{step:chooserep}
\For   {each node $w \in A-\{v_A\}$} 
\State Shortcut the flow in $F$ over $w$ so $F(\delta^+(w)) = 0$ \Comment $F$ remains acyclic \label{step:shortcut}
\EndFor
\State $W \leftarrow W - (A - \{v_A\})$ \label{step:remove} \Comment The nodes in $A-\{v_A\}$ will be reintroduced in step \ref{step:circ}
\State $l_{v_A} \leftarrow l_{v_A} + H(\delta^+_A(v_A))$
\EndFor
\State $C \leftarrow C + H$ \label{step:addcirc}
\EndFor 
\State Use Corollary \ref{cor:findpaths} to find $k' \in [k, k+\frac{k}{b}]$ $s-t$ paths whose union covers all nodes in $W$ of total cost at most the cost of $F$. Let $P$ be the $s-t$ flow of value $k'$ defined by these paths. \label{step:wflow}
\State Let $B$ be the simple $t-s$ flow with $B_{ts} = k'$ and $B_{uv} = 0$ for every other arc.
\State Let $C' \leftarrow P+C+B$ \label{step:circ}
\State Find an Eulerian circuit $K$ using each arc $uv$ eactly $C'_{uv}$ times. \Comment See Lemma \ref{lem:circsupport}
\State Shortcut $K$ so it visits each node in $V-\{s,t\}$ exactly once.
\State Delete the $k'$ edges $ts$ from $K$ to obtain $k'$ paths $P_1, \ldots, P_{k'}$. \label{step:kpath}
\State Return $P_1, \ldots, P_{k'}$. \label{step:return}
\end{algorithmic}
\end{algorithm}

The proof of the following lemma is simple and is found in Appendix \ref{app:circ}.
\begin{lemma} \label{lem:circsupport}
In Step \ref{step:circ}, $C'$ is an integral circulation whose support is strongly connected in $G(V,A)$.
\end{lemma}

We assume, for now, that Step \ref{step:wflow} works as stated. If so, we can prove the following combinatorial analog of Theorem \ref{thm:katspp} that compares the cost of the solution to the optimum $k$-ATSPP cost $OPT$ solution rather than the value of LP (\ref{lp:katspp}).
\begin{theorem} \label{thm:comb_opt}
Each node $v \in V-\{s,t\}$ lies on exactly one of the $s-t$ paths $P_1, \ldots, P_{k'}$ returned by Algorithm \ref{alg:katspp} and the total cost of these paths is at most $L \cdot OPT$.
\end{theorem}
\begin{proof}
Since $C'$ is a circulation whose support is strongly connected in $G(V,A)$, then every node in $V$ is visited by the Eulerian circuit. Since the shortcutting did not bypass around either $s$ or $t$, then the $ts$ arc still appear exactly $k'$ times in $K$ and there are no $vs$ or $tv$ arcs in $K$ for any $v \in V-\{s,t\}$. So, after removing the $k'$ occurences of the $t-s$ arc from $K$, we have a collection of precisely $k'$ paths from $s$ to $t$ whose union visits all nodes.

All that is left is to bound the final cost by $L \cdot OPT$. We prove, by induction on $i$, that the cost of $F+C$ after iteration $i$ is at most $i \cdot OPT$. For $i = 0$ (before the first iteration) this is clear and we now assume that $i > 0$ and that the cost of $F+C$ just before the $i$'th iteration is at most $(i-1) \cdot OPT$.

Let $W_i$ be the subset of nodes $W$ in the $i$'th iteration. We can obtain a feasible $k$-path/cycle cover on $W_i$ of cost at most $OPT$ by shortcutting an optimum $k$-ATSPP solution on $V$ around nodes in $V-W_i$. Thus, the minimum cost $k$-path/cycle cover $F'$ on $W_i$ has cost at most $OPT$. After adding $F'$ to $F$ we have that the cost of $F+C$ is, by induction, at most $i \cdot OPT$. The rest of the body of the loop simply moves flow between $F$ and $C$ and shortcuts some flow so the cost of $F+C$ is still bound by $i \cdot OPT$ at the end of the $i$'th iteration.

The cost of the circulation $C'$ in Step \ref{step:circ} is then at most $L \cdot OPT$ plus the cost of $B$. Shortcutting past nodes in the Eulerian circuit $K$ yields an circuit whose total cost is still at most $L \cdot OPT$ plus the cost of $B$. The paths $P_1, \ldots, P_{k'}$ are formed by removing the $k'$ arcs from $t$ to $s$ which have cost exactly the cost of $B$. Thus, the cost of the returned paths is at most $L \cdot OPT$.
\end{proof}

\subsection{Finding the Flow $P$ in Step \ref{step:wflow}}

Consider the acyclic integral flow $F$ after the main loop terminates. It is possible to argue that our choice of $v_A$ in Step \ref{step:chooserep} implies each $v \in W$ has $F(\delta^+(v)) > 0$ so a path decomposition of $F$ yields a collection of paths whose union covers all nodes. However, the number of paths is $kL$ which is much larger than our desired value $k + \frac{k}{b}$. The fact that we can find fewer paths is essentially due to the fact that every $v \in W-\{s,t\}$ supports a lot of flow in $F$.

The main object of concern in this step is the following polytope $\mathcal P(D)$ where $D \in \mathbb R$. In $\mathcal P(D)$, we have a variable $z_{uv}$ for every arc $uv$ in the subgraph induced by $W$. The full description of $\mathcal P(D)$ is:
\begin{align}
  ~~~z(\delta^+(w)) = z(\delta^-(w)) & = ~~ 1 & \forall w \in W - \{s,t\} \label{lp:pd} \\
 z(\delta^+(s)) = z(\delta^-(t)) & = ~~ D & \label{lp:pd-d} \\
 z(\delta^-(s)) = z(\delta^+(t)) & = ~~ 0 & \label{lp:pd-zero} \\
 0 \leq z_{uv} & \leq ~~ F_{uv} & \forall {\rm ~ordered~pairs~} u,v \in V \label{lp:pd-bound}
\end{align}

Since the support of $F$ is acyclic and the support of $z$ is required to be a subset of the support of $F$, then any integral point $z \in \mathcal P(D)$ corresponds to a flow $P$ of the form required in Step \ref{step:wflow} with $k' = D$. Notice that Constraints (\ref{lp:pd-bound}) imply that the cost of $z$ would be at most the cost of $F$. Thus, our goal is to find a value $D \in [k,k+\frac{k}{b}]$ for which $\mathcal P(D)$ has an integer point. The proof of the following property is standard ({\em eg.} \cite{schrijver}).

\begin{lemma}
Every basic point in polytope $\mathcal P(D)$ is integral when $D$ is an integer.
\end{lemma}

So, to prove $\mathcal P(D)$ has an integer point for a given integer $D$ it suffices to prove that $\mathcal P(D)$ contains {\em any} point. That is, for $D \in \mathbb Z$ if there is some point $z \in \mathcal P(D)$ with, perhaps, rational coordinates, then there is certainly a point $z'$ with integer coordinates since $\mathcal P(D)$ is integral.

The following lemmas are proven in essentially the same way as Lemmas 2.3 and 2.5 in the extended abstract of \cite{friggstad:salavatipour:svitkina}, so we omit their proofs.

\begin{lemma} \label{lem:bound}
Throughout the course of the algorithm, $l_v \leq \lfloor \log_2 n\rfloor$ for every $v \in W$.
\end{lemma}

\begin{lemma} \label{lem:support}
When the main loop terminates, $F(\delta^+(v)) = L-l_v$ for each $v \in W - \{s,t\}$.
\end{lemma}

We will require that all node in $W - \{s,t\}$ support the same amount of flow in $F$. The following lemma shows how to construct such a flow through simple modifications of $F$.
\begin{lemma} \label{lem:uniform}
For some $0 \leq \gamma \leq \lfloor \log_2 n \rfloor$ there is an acyclic integral flow $F'$ sending $kL$ units of flow from $s$ to $t$ where every $v \in W-\{s,t\}$ has $F'(\delta^+(v)) =  L - \gamma$. Furthermore, the cost of $F'$ is at most the cost of $F$.
\end{lemma}
\begin{proof}
Let $\gamma = \max_{w \in W-\{s,t\}} l_w$ and recall that $\gamma \leq \lfloor \log_2 n \rfloor$ by Lemma \ref{lem:bound}. While some $w \in W - \{s,t\}$ has $F(\delta^+(v)) > L - \gamma$ ({\em cf.} Lemma \ref{lem:support}), shortcut $F$ past $w$ as follows. Choose any two arcs $uw,wv$ with $F_{uw},F_{wv} > 0$. Then subtract 1 from both $F_{uw}$ and $F_{wv}$ and add 1 to $F_{uv}$. Note that $F$ remains integral and acyclic after such an operation and that the cost of $F$ does not increase by the triangle inequality. We let $F'$ be this flow $F$ once all $w \in W-\{s,t\}$ have $F(\delta^+(v)) = L - \gamma$.
\end{proof}

From now on, we will assume that $F(\delta^+(v)) = L - \gamma$ for every $v \in W-\{s,t\}$ where $\gamma$ is some integer at most $\log_2 n$. The following lemma is the first step to finding a good integer $D$ for which $\mathcal P(D) \neq \emptyset$.
The value $\frac{kL}{L-\gamma}$ may be fractional, but we will deal with that problem later.
\begin{lemma} \label{lem:polypoint}
$\mathcal P\left(\frac{kL}{L-\gamma}\right) \neq \emptyset$ for some $0 \leq \gamma \leq \lfloor \log_2 n \rfloor$.
\end{lemma}
\begin{proof}
Let $D = \frac{kL}{L-\gamma}$ and define a point $z$ by $z_{uv} = \frac{F_{uv}}{L-\gamma}$. It is easy to verify $z \in \mathcal P(D)$ after noting that $\gamma \leq \lfloor \log_2 n \rfloor$ and $b \geq 1$ implies $L-\gamma \geq 1$.
\end{proof}

The main problem is that the $\frac{kL}{L-\gamma}$ may not be an integer. The following lemma fixes this by mapping a point in $\mathcal P(D)$ to a point in $\mathcal P(\lfloor D \rfloor)$ be modifying flow along ``sawtooth'' $s-t$ paths that alternate between following arcs in the support of $z$ in the forward and reverse directions. This can be used to decrease both $z(\delta^+(s))$ and $z(\delta^-(t))$ while preserving $z(\delta^+(v)) = z(\delta^-(v)) = 1$ at all other nodes. The fact that such a sawtooth path always exists if $D$ is not an integer is a consequence of the fact that the total flow through all nodes in $W-\{s,t\}$ is integral. The full proof appears in Appendix \ref{app:pdpoint}.

\begin{lemma} \label{lem:sawtooth}
If $\mathcal P(D) \neq \emptyset$, then $\mathcal P(\lfloor D \rfloor) \neq \emptyset$.
\end{lemma}

\begin{corollary} \label{cor:findpaths}
There is an integer $k' \in [k,k + \frac{k}{b}]$ such that $\mathcal P(k') \neq \emptyset$. That is, we we can find between $k$ and $k+\frac{k}{b}$ paths from $s$ to $t$ whose union spans all nodes in $W$. Furthermore, the cost of these paths is at most the cost of the flow $F$ after the main loop of Algorithm \ref{alg:katspp}.
\end{corollary}
\begin{proof}
Lemmas \ref{lem:polypoint} and \ref{lem:sawtooth} show $\mathcal P\left(\left\lfloor \frac{kL}{L-\gamma} \right\rfloor\right) \neq \emptyset$ for some $0 \leq \gamma \leq \lfloor \log_2 \rfloor n$. The result follows since
\[ k \leq \left\lfloor \frac{kL}{L-\gamma}\right\rfloor \leq \frac{k(b+1)\lfloor \log_2n \rfloor}{b \lfloor \log_2 n \rfloor} = k + \frac{k}{b}. \]
\end{proof}

In the next section, we complete the proof of Theorem \ref{thm:katspp} by showing the cost is actually at most $O(b \log n)$ times the value of LP relaxation (\ref{lp:katspp}).


\subsection{Bounding the Integrality Gap}

For a subset $W \subseteq V$ containing both $s$ and $t$, let $LP(W)$ denote LP relaxation (\ref{lp:katspp}) on the asymmetric metric induced by $W$. Consider the LP obtained by removing Constraints (\ref{lp:katspp-set}) from $LP(W)$. The resulting LP is integral \cite{schrijver} whose integer points correspond to $k$-path/cycle covers of $W$. The following holds because removing the constraints from a minimizaton LP does not increase its value.

\begin{lemma}\label{lem:katspp_lp}
For any subset $W \subseteq V$ containing $s$ and $t$, the minimum cost of a $k$-path/cycle cover of $W$ is at most the value of $LP(W)$.
\end{lemma}

The proof of the next lemma is the same as the proof for the case $k=1$ in \cite{friggstad:salavatipour:svitkina}. It uses splitting off techniques developed by Frank \cite{frank} and Jackson \cite{jackson} for Eulerian graphs. The idea is that we obtain an Eulerian graph by multiplying an optimum basic (thus fractional) point in LP (\ref{lp:katspp}) by a large enough integer and identifying $s$ and $t$. Since we are only concerned with preserving cuts for subsets $S$ of $W$ that include $s$, then using splitting off techniques to bypass nodes in $V-\{s,t\}$ in this Eulerian graph and then scaling back to a fractional point $x'$ will preserve $x'(\delta^+(S)) \geq 1$ in the graph induced by $W$ while not increasing the cost.
\begin{lemma} \label{lem:lp_restrict}
For any subset $W$ of $V$ containing $s$ and $t$, the value of $LP(W)$ is at most the value $OPT_{LP}$ of $LP(V)$.
\end{lemma}

Now we can complete the proof of Theorem \ref{thm:katspp}.

\vskip 2mm
\begin{proof}[Proof of Theorem \ref{thm:katspp}.]
By Lemmas \ref{lem:katspp_lp} and \ref{lem:lp_restrict}, the cost of each $k$-path/cycle cover found in Step (\ref{step:kpathcycle}) is at most $OPT_{LP}$. The proof of Theorem \ref{thm:comb_opt} shows that the final cost of the paths is at most the sum of the costs of each $k$-path/cycle cover found. So, the total cost of the $k + \frac{k}{b}$ paths found is at most $L \cdot OPT_{LP} = O(b \log n) \cdot OPT_{LP}$.
\end{proof}


\subsection{Varying the Endpoints} \label{sec:vary}

Consider the following ways to specify the start locations of the paths:
each path may start at any node (No Source),
all paths start at a common node $s$ (Common Source),
or there are nodes $s_1, \ldots, s_k$ where each must be the start of some path (Multiple Source).
We can also consider analogous ways to specify the end locations of the paths. In Multiple Source/Multiple Sink instances, we only require each path start at some $s_i$ and end at some $t_j$. It may be that some paths start and end at locations with different indices. 

The following theorems are easy to verify and the proofs are only briefly sketched. We let $OPT$ be the cost of the optimum solution using exactly $k$ paths for the $k$-ATSPP variant in question.
\begin{theorem} \label{thm:noend}
For any integer $b \geq 1$, there is an approximation algorithm for the No Source/Single Sink variant of $k$-ATSPP that finds at most $k+\frac{k}{b}$ paths of total cost at most $O(b \log n) \cdot OPT$.
\end{theorem}
\begin{proof}
Simply add a new start node $s$ and set $sv = 0$ and $vs = \infty$ for every $v \in V$. Then use the approximation algorithm from Theorem \ref{thm:katspp}.
\end{proof}

\begin{theorem} \label{thm:varyend}
For any integer $b \geq 1$, there is an approximation algorithm for the Multiple Source/Single Sink variant of $k$-ATSPP that finds at most $k+\frac{k}{b}$ paths of total cost at most $O(b \log (n+k)) \cdot OPT$.
\end{theorem}
\begin{proof}
Let $s_1, \ldots, s_k$ be the multiple sources. Create a start node $s$ and $k$ other new nodes $s'_1, \ldots, s'_k$. Add cost 0 arcs from $s$ to $s'_i$ and from $s'_i$ to $s_i$ for each $i$. Add a cost $\infty$ arc from $v$ to $s$ for every $v \in V$. Use Theorem \ref{thm:katspp} on the shortest paths metric of this graph. The intermediate nodes $s'_i$ are to ensure that each $s_i$ is the start location of some path.
\end{proof}

\vskip 2mm

Combining the constructions in Theorems \ref{thm:noend} and \ref{thm:varyend} can also be used to combine different start and end location specifications ({\em e.g.} No Source/Multiple Sink).


\section{General $k$-ATSPP} \label{sec:general_katspp}
Recall that in General $k$-ATSPP, we are given $k$ pairs of nodes $(s_1,t_1), \ldots, (s_k,t_k)$. The goal is to find an $s_i-t_i$ path for each $1 \leq i \leq k$ so that every $v \in V$ lies on at least one such path. This differs from the Multiple Source/Multiple Sink variant described in Section \ref{sec:vary} since the path at $s_i$ must end at $t_i$, rather than merely requiring that the start and end nodes of the paths establish a bijection between $\{s_1, \ldots, s_k\}$ and $\{t_1, \ldots, t_k\}$.

\subsection{Approximating General 2-ATSPP}
Let $(s_1,t_1)$ and $(s_2,t_2)$ be pairs of nodes we are to connect. Furthermore, suppose all four of these endpoints are distinct (by creating multiple copies of locations if necessary). This means there is an optimum General 2-ATSPP solution where the two paths are vertex disjoint.

Notice that the optimum Multiple Source/Multiple Sink 2-ATSPP solution for the case with sources $\{s_1,s_2\}$ and sinks $\{t_1, t_2\}$ is a lower bound for $OPT$. The first step of the algorithm is to run an $\alpha$-approximation for this variant of 2-ATSPP. This gives us two paths $P_1, P_2$ starting at $s_1, s_2$, respectively. If $P_1$ ends at $t_1$ (equivalently, $P_2$ ends at $t_2$), then these paths form a valid solution for the General 2-ATSPP problem with cost at most $\alpha \cdot OPT$. Otherwise, we have $s_1-t_2$ and $s_2-t_1$ paths. We use the following lemma which is proven in Appendix \ref{app:splice}. It gives us a relatively cheap way to adjust these paths to get $s_1-t_1$ and $s_2-t_2$ paths.

\begin{lemma} \label{lem:splice}
There are nodes $u_1,v_2$ on $P_1$ and $v_1,u_2$ on $P_2$ such that the following hold: a) $u_1=v_2$ or $u_1v_2$ is an arc on $P_1$, b) $v_1=u_2$ or $v_1u_2$ is an arc on $P_2$ and c) $d_{u_1u_2} + d_{v_1v_2} \leq OPT$.
\end{lemma}

The rest of the General 2-ATSPP algorithm proceeds as follows. Try all $O(n^2)$ guesses for $u_1,u_2,v_1,v_2$ where either $u_1=v_2$ or $u_1v_2$ is an arc on $P_1$ and either $v_1=u_2$ or $v_1u_2$ is an arc on $P_2$. For each guess, construct a path $Q_1$ by traveling along $P_1$ from $s_1$ to $u_1$, then using the $u_1u_2$ arc in $G$, and then traveling along $P_2$ from $u_2$ to $t_1$.  Construct $Q_2$ in a similar manner by traveling from $s_2$ to $v_1$ on $P_2$, then using the $v_1v_2$ arc in $G$, and then traveling along $P_1$ from $v_2$ to $t_2$. It is easy to see that $Q_1$ and $Q_2$ form a feasible solution for this General 2-ATSPP instance. Since each arc on $P_1$ and $P_2$ is traversed at most once by $Q_1$ and $Q_2$, then the total cost of $Q_1$ and $Q_2$ is at most $\alpha \cdot OPT + d_{u_1u_2} + d_{v_1v_2}$. Output the cheapest solution found over these guesses.

When the algorithm guesses nodes $u_1, u_2, v_1, v_2$ from Lemma \ref{lem:splice} we have $d_{u_1u_2} + d_{v_1v_2} \leq OPT$. It is straightforward to verify that each arc in $P_1$ and $P_2$ is used at most once between $Q_1$ and $Q_2$. So the final cost of $Q_1$ and $Q_2$ is at most $(\alpha+1)\cdot OPT$. To summarize:

\begin{theorem} \label{thm:general_2atspp}
If there is an $\alpha$-approximation for the Multiple Source/Multiple Sink variant of 2-ATSPP then there is an $(\alpha + 1)$-approximation for General 2-ATSPP.
\end{theorem}

By setting $b=3$ and using Theorem \ref{thm:varyend} composed with an analogous result for Multiple Sink instances, Theorem \ref{thm:general_2atspp} gives us the following.

\begin{corollary}
There is an $O(\log n)$-approximation for General 2-ATSPP.
\end{corollary}


\subsection{Approximating Other Restrictions of General $k$-ATSPP}

We breifly mention a couple of variants of General $k$-ATSPP that can be approximated well. We leave their full descriptions to Appendix \ref{app:general_katspp} since they are simple variants of known algorithms for some TSP variants.

The first variant is when the metric is symmetric. In this case, there is a simple 3-approximation using a tree doubling approach. Next, if $s_i = t_i$ for each $1 \leq i \leq k$, then the problem is to find a cycle cover where each cycle contains some root node $s_i$ where we allow cost 0 loops on the nodes $s_i$ (corresponding to the salesman at $s_i$ going directly to $t_i$). This version can be approximated within $\lfloor\log_2 (n - k)\rfloor + 1$ using a modification of the ATSP algorithm by Frieze {\em et al.} \cite{frieze:galbiati:maffioli}.


\subsection{Hardness of General $k$-ATSPP}

We will use the following NP-complete.

\begin{definition}
In the Tripartite Triangle Packing problem,
we are given a tripartite graph $G = (U \cup V \cup W, E)$ with $|U| = |V| = |W| = n$ where no edge in $E$ has both endpoints in a common partition $U,V$, or $W$. A triangle is a subset of 3 nodes for which any two are adjacent in $G$. The problem is to determine if it is possible to find $n$ vertex-disjoint triangles in $G$.
\end{definition}

NP-completeness of this problem is essentially shown in \cite{garey:johnson}. Technically, only related problems are proven to be NP-complete in this book. However, consider the 3D Matching problem that is proven to be NP-complete in Theorem 3.2 on page 50. If each triple $M = \{u,v,w\}$ is replaced with edges $uv,vw,wu$ then we obtain an instance of Tripartite Triangle Packing. Careful inspection of the proof of Theorem 3.2 of \cite{garey:johnson} shows that the only triangles $\{uv,vw,wu\}$ that can be formed must have come from some triple $M = \{u,v,w\}$ in the 3D Matching reduction. Thus, the Tripartite Triangle Packing problem is also NP-complete.

For our reduction to General $k$-ATSPP, let $G = (U \cup V \cup W, E)$ be an instance of Tripartite Triangle Packing with $|U| = |V| = |W| = n$. Create a directed graph $H$ with four layers of nodes $X_1, X_2, X_3, X_4$ where $X_1$ and $X_4$ are disjoint copies of $U$, $X_2$ is a copy of $V$, and $X_3$ is a copy of $W$. For every edge $e$ in $G$, there is a unique index $1 \leq i \leq 3$ such that the endpoints of $e$ lie in $X_i$ and $X_{i+1}$. Add this arc to $H$,  direct it from $X_i$ to $X_{i+1}$, and set its cost to 0. This is illustrated in Figure \ref{fig:reduction}. Set $k := n$ and consider the General $k$-ATSPP instance on $H'$ obtained from the shortest paths metric $H$ where we set the cost of a $uv$ arc to be 1 if there is no $u-v$ path in $H$. For each $u \in U$, we have a source/sink pair from the copy of $u$ in $X_1$ to the copy of $u$ in $X_4$.

\begin{figure}
\centering
\includegraphics[scale=0.7]{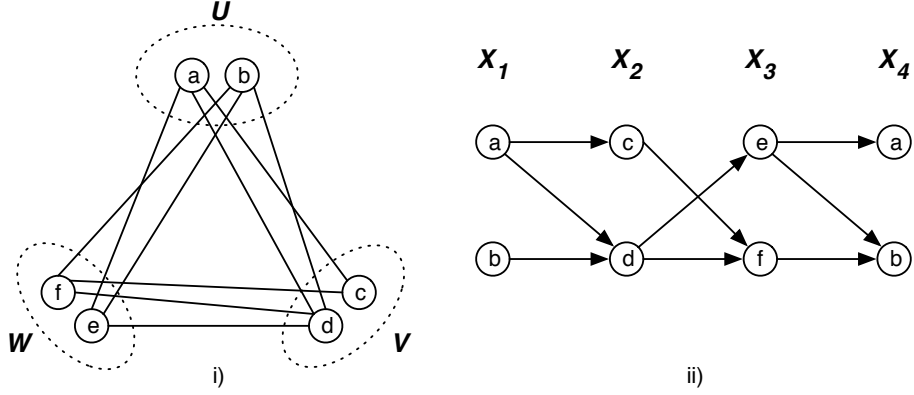}
\caption{
i) An instance of tripartite triangle packing with $n = 2$. ii) The corresponding graph $H$.
}
\label{fig:reduction}
\end{figure}

The details of the following claim are simple and can be found in Appendix \ref{app:reduce}. The proof of Theorem \ref{thm:general_hard} immediately follows.

\begin{claim} \label{claim:hard}
There is a Tripartite Triangle Packing solution in $G$ if and only if the optimum General $k$-ATSPP solution in $H'$ solution has cost 0.
\end{claim}

Note that the value $k$ is $n/4$ in the above reduction. Similar hardness results for smaller values of $k$ (as a function of $n$) are established in Appendix \ref{app:smallk} down to $k$ being polylogarithmic in $n$. The complexity of approximating the case when $k$ is a constant at least 3 remains open.


\section{Future Directions} \label{sec:conc}

Our best approximation for $k$-ATSPP that uses exactly $k$ paths has an approximation guarantee of $O(k \log n)$. Can the dependence on $k$ in the approximation ratio be reduced? Perhaps there is an $O({\rm polylog}(n,k))$-approximation for $k$-ATSPP that uses only $k$ paths. On the other hand, the problem might be hard to approximate much better than $k$. Also, as far as we know the integrality gap of LP (\ref{lp:katspp}) could be $\Omega(k)$.

For General $k$-ATSPP, the case $k=1$ is simply ATSPP and we proved an $O(\log n)$-approximation for $k=2$. Is there a more general $O(f(k) \cdot \log n)$-approximation for General $k$-ATSPP whose running time is polynomial when $k$ is a constant?

Finally, rather than minimizing the total cost of all paths we might want to minimize the cost of the most expensive path. This can be thought of as minimizing the total time it takes for agents moving simultaneously to visit all locations. From Theorem \ref{thm:katspp} and the observation that the total cost of $k$ paths is at most $k$ times the cost of the most expensive path, we get an $O(bk \log n)$-approximation for this variant that uses at most $k + \frac{k}{b}$ paths. Also, Theorems \ref{thm:orient} and \ref{thm:kstroll} (Appendix \ref{sec:makespan}) show that approximation algorithms for Directed Orienteering and Directed $k$-Stroll can be used to obtain other bicriteria approximation algorithms. The current best approximation algorithms for Directed Orienteering \cite{chekuri:korula:pal, nagarajan:ravi:polylog} and Directed $k$-Stroll  \cite{bateni:chuzhoy} imply the following:
\begin{corollary}
There is an efficient algorithm that finds $O(k \log^3 n)$ paths from $s$ to $t$ of maximum cost $OPT$ whose union covers all nodes in $V$.
\end{corollary}
\begin{corollary}
There is an $O(\log^3 n)$-approximation uses at most $O(k \log n)$ paths.
\end{corollary}
We leave it as an open problem to improve these bounds. In particular, is it possible to obtain a polylogarithmic approximation that uses only $O(k)$ paths? We also note that the hardness results for General $k$-ATSPP proven in this paper also hold for the the variant where we want to minimize the maximum cost of the $s_i-t_i$ paths.

\section{Acknowledgements}
The author would like to thank Anupam Gupta, Mohammad R. Salavatipour, and Zoya Svitkina for insightful discussions on these problems.


\bibliography{katspp.bbl}


\appendix

\section{Appendix}

\subsection{Proof of Lemma \ref{lem:circsupport}} \label{app:circ}

\begin{proof} [~]
We see that $C'$ is an integral circulation since it is the sum of an integral circulation $C$, an integral $s-t$ flow $P$ of value $k'$, and an integral $t-s$ flow $B$ of value $k'$. To prove its support is strongly connected in $G(V,A)$ it suffices to show that there is a $v-t$ path in the support of $C'$ for every $v \in V$. Suppose $W_i$ is the set of remaining nodes $W$ just after the $i$'th iteration of the loop (and $W_0 = V$). We prove by induction on $L-i$ that every node in $W_i$ has a path to $v$.

Consider the base case $i = L$. Since $P$ only uses arcs in the support of $F$, since $F$ is acyclic (from Step \ref{step:circulation}), and since $F(\delta^+(v)) > 0$ for $v \in W-\{s,t\}$ (Lemmas \ref{lem:bound} and \ref{lem:support}) then any walk from any $v \in W_L$ in the support of $P$ of maximal length must end at $t$. 

Now suppose $i < L$ and that every node in $W_{i+1}$ has a path to $t$ in the support of $C'$. Let $v \in W_i$, if $v \in W_{i+1}$ then, by induction, $v$ has a path to $t$. Otherwise, $v$ must have been removed in Step \ref{step:remove} of the current iteration for some circulation $A$. But then $v_A$, the representative of circulation $A$ that was chosen to remain in $W$, is in $W_{i+1}$, so it has a path to $t$ by induction. Finally, $v$ also has a path to $t$ in the support of $C'$ since it can travel first to $v_A$ in $A$ and then, by induction, to $t$.
\end{proof}


\subsection{Proof of Lemma \ref{lem:sawtooth}} \label{app:pdpoint}
\begin{proof}[~]
Suppose $D$ is not an integer, otherwise we are already done. Let $z$ be any point in $\mathcal P(D)$. Form an undirected and weighted bipartite graph $H = (W_L \cup W_R, E')$ where both $W_L$ and $W_R$ are disjoint copies of $W$. We identify an arc $uv$ of $G$ with an edge $uv$ in $H$ with $u \in W_L$ and $v \in W_R$. That is, for each arc $uv$ add an edge from $u \in W_L$ to $v \in W_R$ with weight $z_{uv}$. By constraints (\ref{lp:pd}), (\ref{lp:pd-d}) and (\ref{lp:pd-zero}), we have that $z(\delta(v))$ (the total $z$-value of all edges in $E'$ incident to $v$) is an integer for every $v \in W_L \cup W_R$ except for $s \in W_L$ and $t \in W_R$.

We claim there is a path from $s \in W_L$ to $t \in W_R$ in $H$ that only uses edges $uv$ with $z_{uv} > 0$. Note that these paths are allowed to take a step from $W_R$ to $W_L$ since $H$ is undirected. Such a step corresponds to following an arc $uv$ in the reverse direction in the original directed graph $G$.

Suppose, for the sake of contradiction, that there is no path from $s \in W_L$ to $t \in W_R$ using only edges $uv$ with $z_{uv} > 0$. Let $S$ be the collection of all nodes in $H$ that can be reached from $s \in W_L$ using only edges with positive $z$-value; our assumptions means the copy of $t$ in $W_R$ is not in $S$.

Let $z(E(S))$ denote the total $z$-value of edges with both endpoints in $S$. On one hand, since every node $v \in W_R-\{t\}$ has $z(\delta(v)) = 1$ and since $t \not\in S$, then 
\[z(E(S)) = \sum_{v \in W_R \cap S} z(\delta(v)) = |W_R \cap S|.\]

On the other hand, since every node $v \in W_L-\{s\}$ has $z(\delta(v)) = 1$ and since $s \in S$, then 
\[z(E(S)) = \sum_{v \in W_L \cap S} z(\delta(v)) = |(W_L-\{s\}) \cap S| + z(\delta(s)) = |(W_L - \{s\}) \cap S| + D.\]
But $|R \cap S| = |(W_L - \{s\}) \cap S| + D$ contradicts our assumption that $D$ is not an integer. So, it must be that there is a path 
from $s \in W_L$ to $t \in W_R$ in $H$ using only edges $e$ with $z_e > 0$.

Suppose that such a path followed a sequence of edges $e_1, e_2, \ldots, e_c$. Since $H$ is bipartite and $s \in W_L, t \in W_R$ are on different sides, then $c$ must be odd.
Let
\[\delta = \min\left\{D - \lfloor D \rfloor, \min_{1 \leq i \leq c : ~ i {\rm ~odd}} z_{e_i}\right\}\]
be the minimum $z$-value of the edges that are followed from $W_L$ to $W_R$ when walking along this path
(or the difference between $D$ and $\lfloor D \rfloor$ if this is smaller).
Update the $z$ values of the edges on this path by:

\[
z_{e_i} \leftarrow \left\{
\begin{array}{ll}
z_{e_i} - \delta & i {\rm ~odd} \\
z_{e_i} + \delta & i {\rm ~even}
\end{array}
\right.
\]

We will now argue that the resulting $z$-values fit in the polytope $\mathcal P(D - \delta)$. First, notice that both $z(\delta(s)), s \in W_L$ and $z(\delta(t)), t \in W_R$, which were originally $D$, decrease by exactly $\delta$. Any other node $v$ not on this path does not have the $z$-value of any incident edge changed. Finally, if $v$ is an internal node on this path then the $z$-value of one incident edge decreases by $\delta$ and the $z$-value of another incident edge increases by $\delta$. Therefore, we have $z(\delta(v)) = 1$ after this update for every internal node $v$ on this path.

By our choice of $\delta$, we maintain $z_e \geq 0$ for every edge $e$. Now, if the path was a single edge $st$ then no edge had its $z$-value increased so the bounds $z_e \leq F_e$ continue to hold for every edge $e$. Otherwise, every edge $e = uv$ on this path has either $z(\delta(u)) = 1$ or $z(\delta(v)) = 1$ so $z_e \leq 1$ must hold after this update. Since the only $z_e$ values that are increased are those in the support of the integral flow $F$, then $z_e \leq 1 \leq F_e$.

The above process maps a point from $\mathcal P(D)$ to a point in $\mathcal P(D - \delta)$ when $D$ is not an integer. If $\delta$ was chosen to be $D - \lfloor D \rfloor > 0$, then $z(\delta(s)) = z(\delta(t)) = \lfloor D \rfloor$ after this process and we are done. Otherwise, we can repeat the process to map a point from $\mathcal P(D - \delta)$ to $\mathcal P(D - \delta')$ for some $\delta' > \delta$ with $D - \delta' \geq \lfloor D \rfloor$ and so on. Each such step that does not result in a point in the polytope $\mathcal P(\lfloor D \rfloor)$ has us remove at least one edge from the support of $z$. Since no edges are introduced to the support of $z$, then this process will terminate in finitely many iterations with a point in $\mathcal P(\lfloor D \rfloor)$.
\end{proof}


\subsection{Proof of Lemma \ref{lem:splice}}\label{app:splice}

\begin{proof}[~]
Let $P'_1$ be the $s_1-t_1$ path and $P'_2$ be the $s_2-t_2$ path in some fixed optimum solution. Also, for nodes $u,v \in V$ we use $u \prec v$ to indicate that both $u$ and $v$ appear on the same path in our optimum solution $P'_1, P'_2$ and that $u$ appears sometime before $v$ on this path.

Let $a,b$ be such that $ab$ is the first arc on $P_1$ with $a \in P'_1$ and $b \in P'_2$. Such an arc exists because $P_1$ starts with a node in $P'_1$ and ends in a node in $P'_2$. Let $x$ be the first node along $P_2$ such that $a \prec x$. We know such a node exists because $a \prec t_1$ and $t_1 \in P_2$. Now, let $y$ be the furthest node along $P_2$ but still before $x$ such that either $y \in P'_1$ or $y \in P'_2$ and $y \prec b$. Again, such a node exists because $s_2 \in P'_2$ and $s_2 \prec b$.

Suppose $y \in P'_2$ with $y \prec b$. If $y$ is immediately followed by $x$ on $P_2$ then we let $u_1 = a, u_2 = x, v_1 = y, v_2 = b$. Otherwise, say $w$ is the immediate successor of $y$ on $P_2$. By our choice of $y$ we have $w \in P'_2$ and $b \prec w$. In this case, we set $u_1 = b, u_2 = w, v_1 = y, v_2 = b$. This case is illustrated in Figure \ref{fig:splice}.

\begin{figure}
\centering
\includegraphics[scale=0.75]{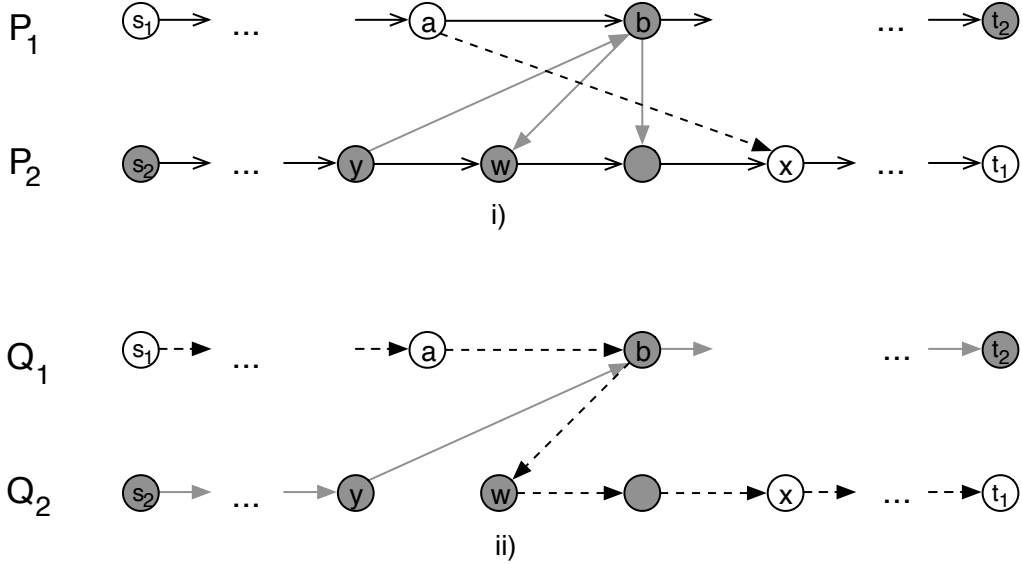}
\caption{
i) An illustration of the case $y \in P'_2$ with immediate successor $w \neq x$ on $P_2$. The white nodes lie on $P'_1$ and the gray nodes lie on $P'_2$. A dashed arc $v \rightarrow v'$ indicates $v \prec v'$ on $P'_1$ and a gray arc $v \rightarrow v'$ indicates $v \prec v'$ on $P'_2$.
ii) The $s_1-t_1$ and $s_2-t_2$ paths $Q_1, Q_2$ formed in the proof of Theorem \ref{thm:general_2atspp} for this case.
}
\label{fig:splice}
\end{figure}

Next, suppose $y \in P'_1$. Let $z$ be the first node on $P_2$ that lies on $P'_1$. Note that $z$ occurs no later than $y$ on $P_2$ (it may be equal to $y$). Now, by our choice of $a$ on $P_1$, every node between $s_1$ and $a$ on $P_1$ also lies on $P'_1$. We also have that $s_1$ appears before $z$ on $P'_1$ (since $s_1$ is the start of $P_1$) and, by our choice of $x$ and the fact that $z$ appears before $x$ on $P_2$, we have that $z$ appears before $a$ on $P'_1$. So, there must be some arc $cd$ on the subpath of $P_1$ starting at $s_1$ and ending at $a$ such that $c$ appears before $z$ on $P'_1$ and $d$ appears after $z$ on $P'_1$. We let $u_1 = c, u_2 = z, v_1 = z, v_2 = d$ in this case.

In all cases, we can easily verify that condition a) and b) in the statement of the lemma are satisfied. Also, in any case we have that one of the following is true:
\begin{enumerate}
\item $u_1 \prec u_2$ and $v_1 \prec v_2$ with $u_1$ and $v_1$ appearing on different paths in $P'_1, P'_2$
\item $v_1 \prec v_2 = u_1 \prec u_2$
\item $u_1 \prec u_2 = v_1 \prec v_2$
\end{enumerate}
In each case, the triangle inequality implies $d_{u_1u_2} + d_{v_1v_2}$ is a lower bound on the total cost of $P'_1$ and $P'_2$. That is, $d_{u_1u_2} + d_{v_1v_2} \leq OPT$.
\end{proof}


\subsection{Restricted Instances of General $k$-ATSPP} \label{app:general_katspp}

\begin{theorem}
There is a 3-approximation for General $k$-ATSPP in symmetric metrics.
\end{theorem}
\begin{proof}
First, we may assume that the $2k$ start and end locations are all distinct by duplicating locations that appear multiple times among $s_1, \ldots, s_k, t_1, \ldots, t_k$. Let $X = \{s_1, \ldots, s_k\} \cup \{t_1, \ldots, t_k\}$.

Compute the minimum weight forest $F$ of $G$ where every component contains exactly one node in $X$. Such forests correspond to spanning trees of $G/X$, the graph obtained by contracting all nodes in $X$ and keeping parallel edges. Now, add the set of edges $M = \{s_it_i, 1 \leq i \leq k\}$ to $F$. This results in a forest where each component contains both $s_i$ and $t_i$ from some $1 \leq i \leq k$ pair. In each such component $F_i$, double all edges except $s_it_i$ and find an Eulerian walk from $s_i$ to $t_i$. Follow such a walk and shortcut past repeated nodes to get a path $P_i$ that visits all nodes in the component $F_i$.

Before adding $M$, the cost of $F$ was at most $OPT$ because a forest where each component contains exactly one node in $X$ can be obtained by deleting one edge from every $s_i-t_i$ path in the optimum solution. Also, since the optimum solution consists of $s_i-t_i$ paths for each $1 \leq i \leq k$ and since the single edge $s_it_i$ is the shortest $s_i-t_i$ path, then the cost of $M$ is also at most $OPT$. Finally, since we only doubled the edges from $F$ and shortcutting does not increase the cost of a path, then the final cost is at most twice the cost of $F$ plus the cost of $M$ which, in turn, is at most $3\cdot OPT$.
\end{proof}

\begin{theorem}
There is a $(\lfloor\log_2 (n - k)\rfloor + 1)$-approximation for instances of General $k$-ATSPP in asymmetric metrics when $s_i = t_i$ for all $1 \leq i \leq k$.
\end{theorem}
\begin{proof}
Let $r_i$ denote the common start and end location for salesman $i$. Recall that we are allowing a salesman to travel directly from $s_i$ to $t_i$ which corresponds to a loop of cost 0 at $r_i$. Like Frieze {\em et al.} \cite{frieze:galbiati:maffioli}, we consider cycle covers except that we also allow a loop at a root $r_i$ to be a cycle in the cycle cover. We do not allow loops at any node in $V-\{r_1, \ldots, r_k\}$. A minimum cost cycle cover that allows loops at root nodes can be computed efficiently in a manner similar to computing minimum cost cycle covers without loops.

Initialize the set of ``remaining nodes'' $W := V$. While $W \neq \{r_1, \ldots, r_k\}$, we compute a minimum cost cycle cover $C$. Note that the cost of $C$ is at most $OPT$ because shortcutting an optimum solution past nodes in $V-W$ yields a feasible cycle cover on $W$ of cost at most $OPT$. If any cycle in $C$ contains more than one root, then shortcut $C$ past all but one of these roots. The roots that are removed from the cycle are then said to be covered by a loop of cost 0. For each cycle $C_i$ in $C$, if there is a root $r$ in $C_i$ then we remove all of $C_i - \{r\}$ from $W$. Otherwise, we only remove all but any one arbitrarily chosen node in $C_i$ from $W$.

When $W$ is finally reduced to $R$, the union of all cycles found over all iterations will have each component being an Eulerian graph containing exactly one root. So, shortcutting an Eulerian circuit in each component yields the desired cycles. At each step before the final iteration, at least half of the non-root nodes are removed so the number of iterations is at most $\lfloor \log_2 (n-k) \rfloor + 1$. The total cost of these cycles is at most the total cost of all cycles covers each of which has cost at most $OPT$, so the total cost is $(\lfloor \log_2 (n-k) \rfloor + 1) \cdot OPT$.
\end{proof}


\subsection{Proof of Claim \ref{claim:hard}} \label{app:reduce}
\begin{proof}[]
Suppose $T = \{(u_i,v_i,w_i)\}_{i=1}^n$ is a collection of $n$ vertex-disjoint triangles with $u_i \in U, v_i \in V$ and $w_i \in W$. For each such triangle $(u_i,v_i,w_i)$, we have the salesman starting at $u_i \in X_1$ to travel first to $v_i \in X_2$, then to $w_i \in X_3$, and finally to $u_i \in X_4$. Every node in $H'$ is visited since every node is included in some triangle in $T$. By construction, every step taken by a salesman traverses an arc of cost 0 so the total cost is 0.

Conversely, suppose there is a $k$-ATSPP solution that avoids using any arc of positive cost. Then each of the arcs followed by the salesmen must be in increasing order of the levels $X_i$. Since there are only $n$ salesmen and since each salesman can visit only two nodes in addition to their endpoints, then each salesman visits every layer $X_i$. Since no arcs of positive cost were used then every salesman only uses arcs corresponding to edges in $G$. Thus, the nodes visited by each salesman correspond to a triangle in $G$ and these triangles partition the nodes of $G$.
\end{proof}


\subsection{Hardness of General $k$-ATSPP for Smaller $k$} \label{app:smallk}
There is a very simple modification to the proof of Theorem \ref{thm:general_hard} that proves hardness for smaller values of $k$. It is reminiscent of some ``padding'' arguments in  complexity theory.

\begin{theorem} \label{thm:smallk}
For any constant $\epsilon > 0$, instances of general $k$-ATSPP with $k = \Omega(n^\epsilon)$ cannot be approximated within any bounded ratio unless $P = NP$.
\end{theorem}
\begin{proof} [Proof (sketch).]
Simply attach a path $P$ of length $\Theta(n^{\frac{1}{\epsilon}})$ to the end of the copy of $u_1 \in X_4$ in $H$ in the reduction from Theorem \ref{thm:general_hard}. Set the cost of all arcs in this path is 0. All start and end locations of the salesmen are the same, except that the first salesman ends at the end of the path $P$ rather than at $u_1 \in X_4$. Now, $k = \Theta(N^{\epsilon})$ where $N$ is the number of nodes in this modified version of $H$.
\end{proof}

In fact, under the Exponential Time Hypothesis we can push the value of $k$ down to polylogarithmic in the hardness reduction.
\begin{theorem} \label{thm:subexp}
There is a constant $c$ such that for any $\delta > 0$, no polynomial-time algorithm for instances of General $k$-ATSPP on $N$ nodes with $k = \Omega(\log^\frac{c}{\delta} N)$ can discern between instances with optimum value 0 and instances with optimum value at least 1 unless 3SAT has can be decided in time $2^{O(n^\delta)}$ where $n$ is the number of variables.
\end{theorem}
\begin{proof}
Say the reduction from 3SAT to Tripartite Triangle Packing produces instances of size at most $n^c$ for large enough $n$ where $c$ is a constant and $n$ is the number of variables in the 3SAT instance. Let $k$ denote the size of a partition in the Tripartite Triangle Packing instance and note $k \leq n^c$. As in the proof of Theorem \ref{thm:smallk}, append a path to some $u_1 \in X_4$ except make its length $2^{k^{\frac{\delta}{c}}}-4k$. Note that $k = (\log_2 N)^\frac{c}{\delta}$ where $N = 2^{k^\frac{\delta}{c}}$ is the number of nodes in this General $k$-ATSPP instance.

Suppose there is a polynomial-time algorithm for General $k$-ATSPP instances on $N$ nodes that can distinguish between instances with optimum value 0 and instances with optimum value at least 1 when $k \geq (\log_2 N)^\frac{c}{\delta}$. Use this algorithm on the instance constructed in the previous paragraph to distinguish between ``yes'' and ``no'' instances of 3SAT.

The running time of the reduction of 3SAT to General $k$-ATSPP is polynomial in $2^{k^{\frac{\delta}{c}}} \leq 2^{n^\delta}$. Running the polynomial-time algorithm for General $k$-ATSPP on the instance with $N = 2^{k^{\frac{\delta}{c}}}$ nodes yields an algorithm for 3SAT whose running time is $2^{O\left(k^{\frac{\delta}{c}}\right)}$ which is at most $2^{O\left(n^\delta\right)}$.
\end{proof}



\subsection{Minimizing the Salesmen's Makespan} \label{sec:makespan}
Here we consider the variant of $k$-ATSPP where the goal is to minimize the cost of the most expensive path rather than the total cost of all paths. Let $OPT$ denote the minimum value for which there are $k$ paths from $s$ to $t$ of maximum length $OPT$ whose union spans all nodes. We first recall that in the Directed Orienteering problem, we have an asymmetric metric $G = (V,A)$ with arc weights, a budget $B$, and a fixed start node $s$. The goal is to visit the maximum number of nodes with a path starting at $s$ of cost at most $B$.
\begin{theorem} \label{thm:orient}
If there is an $\alpha$-approximation for Directed Orienteering, then we can efficiently find $O(k\alpha\log n)$ paths from $s$ to $t$, each of length $2\cdot OPT$, whose union spans all nodes. If we have an $\alpha$-approximation for the variant of Directed Orienteering when both start and end location are specified, then we can find $O(k \alpha \log n)$ paths of length at most $OPT$.
\end{theorem}
\begin{proof}
We will describe an algorithm that is supplied with a value $B$. If $B \geq OPT$, then it will succeed in finding $O(k \alpha \log n)$ paths of length at most $B+OPT$ whose union covers all nodes. So, by binary searching for $B$ and keeping the smallest value for which the algorithm succeeds, we find $O(k \alpha \log n)$ paths of length at most $2 \cdot OPT$.

The algorithm is simple. Initialize $W = V-\{s,t\}$. Then, for $\Theta(k \alpha \log n)$ iterations use the Directed Orienteering approximation on $W\cup\{s\}$ starting from node $s$ to find a path $P$ of length at most $B$ that covers some nodes. Append $t$ to the end of this path $P$. Remove the nodes of $P-\{s,t\}$ from $W$ and repeat. Note that each path has length at most $B+OPT$ because the final step to $t$ has length at most $OPT$. If $W = \emptyset$ after all iterations, then we say that the algorithm succeeded. Otherwise, we say the algorithm failed.

We claim that the algorithm will succeed if $B \geq OPT$. In each iteration we have, from shortcutting the paths in an optimum solution around nodes in $V-W$, that there is a path of length $OPT \leq B$ that covers at least $|W|/k$ nodes in $W$. So, the Directed Orienteering algorithm with budget $B$ will cover at least $|W|/k\alpha$ of these nodes and the size of $W$ drops by a $\left(1 - \frac{1}{k\alpha}\right)$-fraction. After $\Theta(k \alpha \log n)$ iterations all nodes should be covered.

Finally, we note that we can use essentially the same algorithm to find paths of length $OPT$ if we have a $\alpha$-approximation for Directed Orienteering when we can specify both start and end locations. That is, we don't have to extend the paths found by the Directed Orienteering approximation to end at $t$ since they already end at $t$.
\end{proof}

In the Directed $k$-Stroll problem, we have an asymmetric metric $G = (V,A)$, an integer $k$, and two nodes $s,t$. The goal is to find an $s-t$ path of minimum weight that visits at least $k$ other nodes. We can also use approximation algorithms for this problem to approximate the variant of $k$-ATSPP whose goal is to minimize the cost of the most expensive path.
\begin{theorem} \label{thm:kstroll}
If there is an $\alpha$-approximation for the Directed $k$-Stroll problem, then we can find $O(k \log n)$ paths of length at most $\alpha \cdot OPT$ each.
\end{theorem}
\begin{proof}
Initialize $W$ to $V-\{s,t\}$ and iterate the following proceduce. Use the Directed $\lceil |W|/k \rceil$-Stroll approximation on $W \cup \{s,t\}$ to find an $s-t$ path $P$ that visits at least $|W|/k$ locations in $W$. By shortcutting an optimum solution, we see that there is such a path on the subgraph induced by $W$ of length at most $OPT$ so the length of $P$ will be at most $\alpha \cdot OPT$. Remove the nodes in $P-\{s,t\}$ from $W$ and repeat. After $O(k \log n)$ iterations, all nodes will be covered by paths of length at most $\alpha \cdot OPT$.
\end{proof}

\end{document}